\documentclass[journal,twoside,web]{ieeecolor}

\usepackage{generic}
\usepackage{cite}
\usepackage{amsmath,amssymb,amsfonts}
\usepackage{algorithmic}
\usepackage{graphicx}
\usepackage{textcomp}

\usepackage{graphicx}
\usepackage{subfigure}
\usepackage{epsfig} 
\usepackage{cite}
\usepackage{lcsys}

\usepackage{nicematrix}
\usepackage{color}
\usepackage{array}
\usepackage{booktabs}
\usepackage{url}
\usepackage{caption}
\usepackage[
  hidelinks
  ]{hyperref}

\newtheorem{thm}{Theorem}
\newtheorem{lem}{Lemma}
\newtheorem{prob}{Problem}

\newtheorem{assumption}{Assumption}

\newtheorem{rem}{Remark}

\begin{document}


\title{
Design of Distributed Controller for Discrete-Time Systems Via the Integration of Extended LMI and Clique-Wise Decomposition
}

\author{Sotaro Fushimi, Yuto Watanabe, \IEEEmembership{Student Member, IEEE}, and Kazunori Sakurama, \IEEEmembership{Member, IEEE}
\thanks{*This work was supported by JSPS KAKENHI Grant Number 23K22781 and 23K20947.}
\thanks{Sotaro Fushimi is with the Department of Systems Science, Graduate School of Informatics, 
        Kyoto University, Yoshida-Honmachi, Sakyo-ku, Kyoto 606-8501, Japan, {\tt s-fushimi@sys.i.kyoto-u.ac.jp}.}%
\thanks{Yuto Watanabe is with the Department of Electrical and Computer Engineering, 
        University of California San Diego, San Diego, CA 92093 USA, {\tt y1watanabe@ucsd.edu}.}%
\thanks{Kazunori Sakurama is with the Department of System Innovation, Graduate School of Engineering Science, 
        Osaka University, 1-3, Machikaneyama, Toyonaka, Osaka 560-8531, Japan,
        {\tt sakurama.kazunori.es@osaka-u.ac.jp}.}%
}

\maketitle
\thispagestyle{empty}

\begin{abstract}
This study addresses the centralized synthesis of distributed controllers using linear matrix inequalities (LMIs).
Sparsity constraints on control gains of distributed controllers result in conservatism via the convexification of the existing methods such as the extended LMI method.
In order to mitigate the conservatism, we introduce a novel LMI formulation for this problem, utilizing the clique-wise decomposition method from our previous work on continuous-time systems.
By reformulating the sparsity constraint on the gain matrix within cliques, this method achieves a broader solution set.
Also, the analytical superiority of our method is confirmed through numerical examples.
\end{abstract}

\section{INTRODUCTION}
Large-scale systems are generally difficult to control centrally due to the heavy burden on communication and information processing.
Distributed handling of such systems has garnered significant interest and has been actively researched, especially with advancements in sensor and actuator technologies \cite{siljak2013decentralized, ren2008distributed, bullo2009distributed}.

The difficulty in centrally designing a distributed controller arises from the sparsity constraints on the structure of the controller gains.
This results in the non-convexity of optimization problems using linear matrix inequalities (LMIs), whereas in centralized control, problems can be formulated as convex optimization problems.
Exact convexification of this problem is still open, except for special cases such as positive-systems \cite{tanaka2011bounded, rantzer2015scalable}, or for finite horizon \cite{anderson2019system}.
Therefore, several convex relaxation methods have been proposed.
For continuous-time systems, a typical convex relaxation restricts the Lyapunov function to a block-diagonal matrix\cite{siljak1978large, sootla2019existence,vandenberghe2015chordal}. 
For discrete-time systems, alongside the block-diagonal relaxations \cite{siljak1978large, de2000design}, 
another approach based on extended LMIs has been employed in \cite{de1999lmi, de1999new, de2002extended, pipeleers2009extended},
which exploits a specific structure of discrete-time Lyapunov inequalities. 
Although this approach still involves conservatism, it allows us to search for non-block-diagonal Lyapunov matrices.
The conservatism of these methods arises from representing sparse gain matrices as the product of a sparse matrix and a block-diagonal matrix.
To reduce this conservatism, our prior work on continuous-time systems presented a method with clique-wise decomposition \cite{watanabe2024convex}.
However, \cite{watanabe2024convex} still exhibits conservatism, as the Lyapunov matrices are constrained by sparsity.

In this paper, we propose a novel method for designing distributed controllers for discrete-time systems, leveraging both the continuous-time clique-wise decomposition method \cite{watanabe2024convex} and extended LMI techniques \cite{de2002extended, de1999new, pipeleers2009extended}. 
First, we present an explicit formulation of the clique-wise decomposition method based on our previous work on continuous-time systems \cite{watanabe2024convex}. 
Next, we present a new convex LMI formulation, and appliy it to $H_\infty$ controller synthesis. 
Finally, the efficacy of the proposed methods is validated through numerical examples, demonstrating their superiority over existing methods \cite{geromel1994decentralized, de2002extended, watanabe2024convex}.

The main contribution of this paper is the derivation of new, less conservative LMI conditions than existing clique-wise decomposition and extended LMI methods, which cannot be obtained through a straightforward combination of these approaches.
Note that the formulations do not structurally restrict Lyapunov matrices, and such LMI formulations cannot be derived for continuous-time systems, to the best of authors' knowledge.

The paper is organized as follows. 
We present the target system and problem formulation in Section \ref{sec: problem statement}.
Section \ref{sec: preliminaries} provides preliminary lemmas, including the discrete-time formulation of the clique-wise decomposition method \cite{watanabe2024convex}.
Section \ref{sec: main result} presents our main result: the derivation of a new LMI condition. 
Our proposed method is extended to the $H_\infty$ control in Section \ref{sec: Hinf}.
Numerical examples for $H_\infty$ problem are given in Section \ref{sec: numerical examples}.

\textit{Notations:}
$O_{m\times n}\in\mathbb{R}^{m\times n}$ denotes the ${m\times n}$ zero matrix.
${\rm diag}(\dots,a_i,\dots)$ denotes the diagonal matrix with the $i$ th diagonal entry $a_i\in\mathbb{R}$.
Similarly, ${\rm blkdiag}(\dots,A_i,\dots)$ denotes the block-diagonal matrix with the $i$ th matrix entry $A_i$.
For $X\in \mathbb{R}^{m\times n}$ of ${\rm rank}(X)=r<m$, $X^\bot\in\mathbb{R}^{m\times m-r}$ represents the matrix satisfying ${\rm rank}[X,X^\bot]=m$ and $(X^\bot)^\top X=O$.
Non-diagonal components of symmetric matrices are denoted as $\ast$ for brevity.

\section{PROBLEM STATEMENT}\label{sec: problem statement}
Consider the system with
\begin{equation}
        x(k+1)=Ax(k)+Bu(k),\label{eq: sys}
\end{equation}
consisting of $N$ subsystems, where
$x = [x_1^\top,\cdots,x_N^\top]^\top \in \mathbb{R}^n$,
the state of $i$-th subsystem $x_i \in \mathbb{R}^{n_i}$ with
$\sum_{i=1}^{N} n_i = n$,
$u = [u_1^\top,\cdots,u_N^\top]^\top \in \mathbb{R}^m$,
the input of $i$-th subsystem $u_i \in \mathbb{R}^{m_i}$ with
$\sum_{i=1}^{N} m_i = m$, 
$A \in \mathbb{R}^{n\times n}$, and
$B \in \mathbb{R}^{n \times m}$.
We assume that $(A,B)$ is stabilizable.
Suppose that a communication network between subsystems is modeled by a time-invariant undirected graph $\mathcal{G}=(\mathcal{N}, \mathcal{E})$ with $\mathcal{N}=\{1,\dots,N\}$ and an edge set $\mathcal{E}$.
Namely, pairs of different nodes $i,j\in \mathcal{N}$ satisfy $(i,j)\in\mathcal{E}\Leftrightarrow (j,i)\in\mathcal{E}$.
A distributed state feedback controller is defined as
\begin{equation}
        u(k)=Kx(k), \ K\in\mathcal{S},\label{eq: FB}
\end{equation}
where $\mathcal{S}$ represents the set of matrices with graph-induced sparsity defined as
\begin{equation}
        \mathcal{S} = \left\{ K \in \mathbb{R}^{m\times n}: K_{ij}=O_{m_i\times n_j} \ {\rm if} \ (i,j)\notin \mathcal{E}, i \neq j \right\}.\label{eq: sparse}
\end{equation}
For simplicity in the following notation, we assume $n_i = m_i$ for all $i \in \mathcal{N}$ without loss of generality. 
Please refer to Appendix~\ref{App: general notation} for the generalized notation where $n_i \neq m_i$.

The set of all state feedback gains $K$ that stabilize the system \eqref{eq: sys}, denoted by $\mathcal{K}_{\rm all}$, can be characterized as follows:
\begin{equation}
  \begin{aligned}
        \mathcal{K}_{\rm all} =\{& K\in\mathcal{S}: \exists P\succ O\\
        &{\rm{s.t.}} \ 
        (A+BK)^\top P (A+BK) - P \prec 0\}.
  \end{aligned}\label{eq: Kall}
\end{equation}
Conventionally, this problem would be convexified using the Schur complement without the constraints on the distributedness, $\mathcal{K}\in\mathcal{S}$.
This involves transforming the matrix variables $Q=P^{-1}$ and $K=ZQ^{-1}$.
This convexification is ineffective for \eqref{eq: Kall} because $K\in \mathcal{S}$ is transformed into a non-convex constraint, $ZQ^{-1}\in\mathcal{S}$.
To the best of the authors' knowledge, an exact convex formulation of this problem has not yet been found.

A simple way to solve this issue is to impose a block-diagonal structure on the Lyapunov matrices \cite{siljak1978large, de2000design}. 
With this structural constraint, \eqref{eq: Kall} is convexly relaxed into the following solution set, denoted as $\mathcal{K}_{\rm diag}$;
\begin{align}
        \mathcal{K}_{\rm diag} =\{& ZQ^{-1}: \exists Z\in\mathcal{S}, Q={\rm blkdiag} (Q_1,\dots,Q_N)\succ O,\notag\\
        &{\rm{s.t.}} \ 
        \begin{bmatrix}
                Q&\ast\\
                AQ+BZ&Q
        \end{bmatrix}
        \succ O \}.\label{eq: Kdiag}
\end{align}The non-convex constraint $K=ZQ^{-1}\in\mathcal{S}$ is relaxed into a convex constraint by imposing structural constraints $Z\in\mathcal{S}$ and $Q={\rm blkdiag} (Q_1,\dots,Q_N)$.
However, this introduces conservatism, because there are $Z\notin \mathcal{S}$ and non-block-diagonal matrices $Q$ producing $ZQ^{-1}\in\mathcal{S}$.
Note that the Lyapunov matrix $P=Q^{-1}$ is constrained to a block-diagonal structure in this case.

To mitigate this conservatism, extended LMIs \cite{de2002extended, de1999new, pipeleers2009extended} have been proposed.
The following lemma shows that $\mathcal{K}_{\rm ext}$, the solution set of an extended LMI, is a convex relaxation of \eqref{eq: Kall} and less conservative than $\mathcal{K}_{\rm diag}$.
\begin{lem}[Distributed version of \cite{de1999new}]\label{lem: extended lmi}
        Let
        \begin{align}
                \mathcal{K}_{\rm ext} =\{& K=ZG^{-1}\in\mathcal{S}: \exists Z\in\mathcal{S}, Q\succ O,\notag\\
                &G={\rm{blkdiag}}(G_1,\dots, G_N),\notag\\
                &{\rm{s.t.}} \ 
                \begin{bmatrix}
                        G+G^\top-Q&\ast\\
                        AG+BZ&Q
                \end{bmatrix}
                \succ O\}.\label{eq: Kext}
        \end{align} 
        Then, $\mathcal{K}_{\rm diag} \subset\mathcal{K}_{\rm ext}\subset \mathcal{K}_{\rm all}$ holds.
\end{lem}
\begin{proof}
        $\mathcal{K}_{\rm diag} \subset\mathcal{K}_{\rm ext}$ is shown by setting $G=Q$.
        $\mathcal{K}_{\rm ext}\subset \mathcal{K}_{\rm all}$ is shown by transforming Thm. 3 in \cite{de1999new}.
\end{proof}

This formulation removes the block-diagonal constraint on the Lyapunov matrix $P=Q^{-1}$ from $\mathcal{K}_{\rm diag}$ in \eqref{lem: extended lmi} by introducing an auxiliary matrix $G$,
while ensuring $ZG^{-1}\in\mathcal{S}$ by the restrictions $Z\in\mathcal{S}$ and $G={\rm blkdiag} (Q_1,\dots,Q_N)$. 
However, $\mathcal{K}_{\rm ext}$ still brings some degree of conservatism, that is, there is a gap between $\mathcal{K}_{\rm ext}$ and $\mathcal{K}_{\rm all}$.

Another possible way to mitigate the conservatism brought by the structural constraint $ZQ^{-1}\in\mathcal{S}$ is by exploiting the graph structure. 
Our prior work on continuous-time systems \cite{watanabe2024convex} derived a new convex solution set by decomposing the Lyapunov matrix and system matrices in a clique-wise manner. 
A discrete-time version of this formulation can also be derived, as shown in Section~\ref{sec: preliminaries}. 
We denote its solution set as $\mathcal{K}_\mathcal{S}$. 
Although $\mathcal{K}_\mathcal{S}$ is less conservative than $\mathcal{K}_{\rm diag}$, it still exhibits some degree of conservatism, as the Lyapunov matrices are structurally restricted in sparsity.

Note that less conservative solution sets $\mathcal{K}_\mathcal{S}$ and $\mathcal{K}_{\rm ext}$ do not include one another, and we cannot determine which is better.
Our goal in this paper is to find a less conservative convex relaxation of \eqref{eq: Kall} compared to both sets, denoted by $\mathcal{K}_{\mathcal{S}, {\rm ext}}$, as shown in Fig.~\ref{fig: benn}.
In other words, we aim to formulate a convex optimization problem that yields a solution to the following:
\begin{prob}\label{prob 1}
        Consider a linear time-invariant system \eqref{eq: sys} with a static state feedback controller as in \eqref{eq: FB} for an undirected graph $\mathcal{G}$ and the set $\mathcal{S}$ in \eqref{eq: sparse}. 
        Derive a solution set $\mathcal{K}_{\mathcal{S}, {\rm ext}}\subset\mathcal{K}_{\rm all}$ such that $\mathcal{K}_\mathcal{S} \cup \mathcal{K}_{\rm ext}\subset \mathcal{K}_{\mathcal{S}, {\rm ext}}$.
\end{prob}
\begin{figure}[h]
        \begin{center}
                \includegraphics[width = 4.4cm]{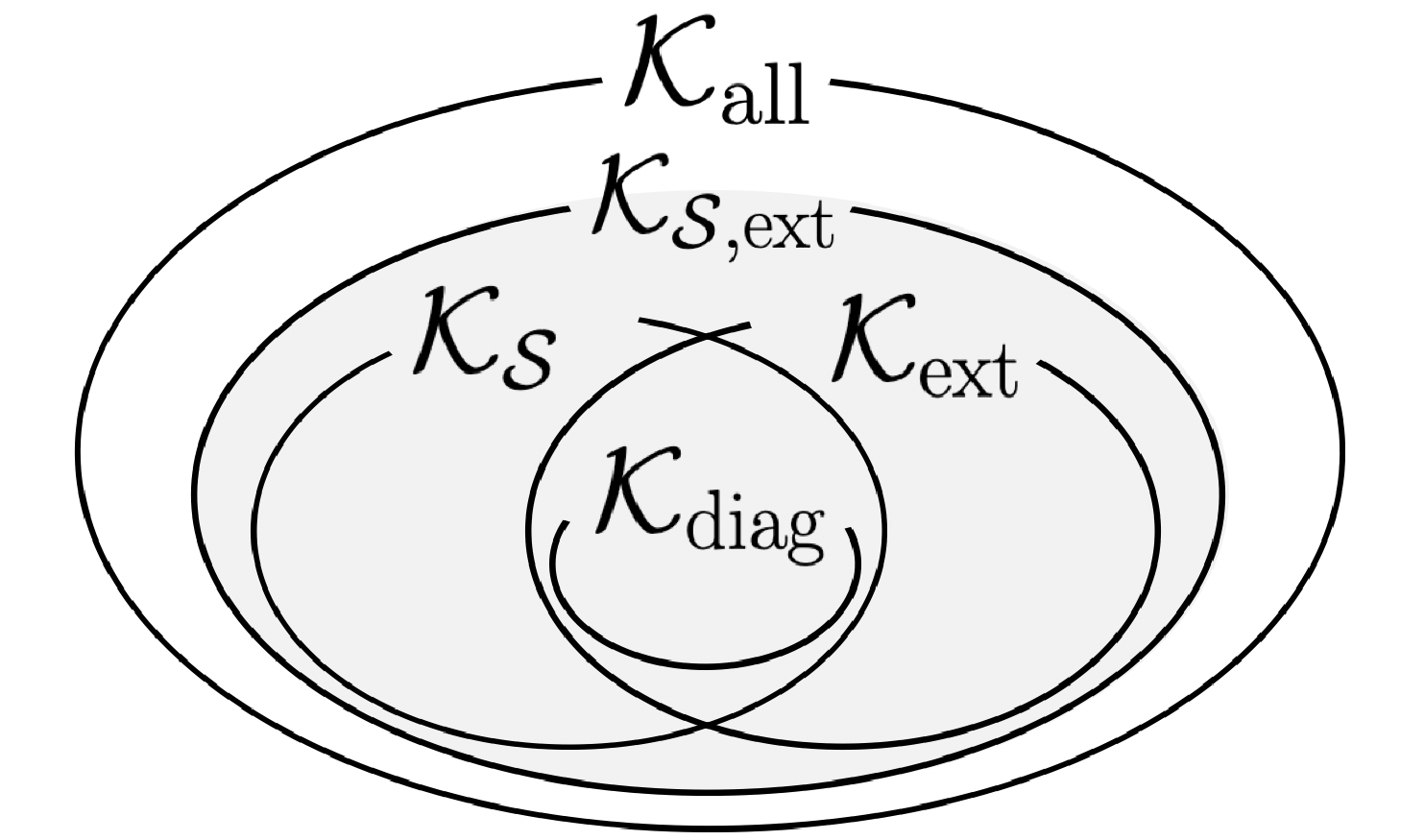}
                \caption{{\small Solution sets}}\label{fig: benn}
        \end{center}
\end{figure}    
\vspace{-0.8em}


\section{PRELIMINARIES}\label{sec: preliminaries}
Consider an undirected graph \(\mathcal{G} = (\mathcal{N}, \mathcal{E})\) for a node set $\mathcal{N}$ and edge set $\mathcal{E}$. 
A clique is a set of nodes that induces a complete subgraph of \(\mathcal{G}\). 
We denote cliques as \(\mathcal{C}_1, \mathcal{C}_2, \dots\) and their index set as \(\mathcal{Q}_\mathcal{G}^{\rm all}\). 
A clique is maximal if it is not included in any other clique. 
The index set for maximal cliques is denoted as \(\mathcal{Q}_\mathcal{G}^{\rm max} \subset \mathcal{Q}_\mathcal{G}^{\rm all}\). 
For a subset \(\mathcal{Q}_\mathcal{G} \subset \mathcal{Q}_\mathcal{G}^{\rm all}\), \(\mathcal{Q}_\mathcal{G}^i\) denotes the subset of cliques in \(\mathcal{Q}_\mathcal{G}\) that contains node \(i\). 
The number of the elements in $\mathcal{Q}_\mathcal{G}^i$ is represented by \(|\mathcal{Q}_\mathcal{G}^i|\).
A graph is said to be chordal if every cycle of length four or more has a chord.

For $\mathcal{Q}_\mathcal{G}$ and clique $\mathcal{C}_k, k\in \mathcal{Q}_\mathcal{G}$, we define matrix $E$ as
\begin{equation}
        \begin{aligned}
                E &= [\dots, E_{\mathcal{C}_k}^\top,\dots]^\top \in \mathbb{R}^{(\Sigma_{k\in\mathcal{Q}_\mathcal{G}} n_{\mathcal{C}_k})\times n},\ k\in\mathcal{Q}_\mathcal{G},\\
                E_{\mathcal{C}_k}&= [\dots, E_j^\top, \dots]^\top \in \mathbb{R}^{n_{\mathcal{C}_k}\times n},\ j\in \mathcal{C}_k,\\
                E_j & = [O_{n_j\times n_1},\dots,I_{n_j},\dots,O_{n_j\times n_N}]\in\mathbb{R}^{n_j\times n},
        \end{aligned}\label{eq: E}
\end{equation}
where $n_{\mathcal{C}_k}=\Sigma_{j\in\mathcal{C}_k} n_j$.
Matrix $E$ satisfies
\begin{itemize}
        \item[a)] $E$ is full column rank.
        \item[b)] $E^\top E={\rm blkdiag}(|\mathcal{Q}_\mathcal{G}^1|I_{n_1},\dots,|\mathcal{Q}_\mathcal{G}^N|I_{n_N})\succ O$.
\end{itemize}

The following factorization of sparse positive difinite matrices is obtained by leveraging the matrix $E$.
\begin{lem}[\cite{watanabe2024convex}] \label{lem: Agler}
        Consider undirected graph $\mathcal{G}=(\mathcal{N}, \mathcal{E})$ with clique index set $\mathcal{Q}_\mathcal{G}=\{1,\dots,q\}$.
        For $\tilde{P}={\rm blkdiag}(\tilde{P}_1,\dots\tilde{P}_q)\succ O$ with $\tilde{P}_k \in \mathbb{R}^{n_{\mathcal{C}_k}\times n_{\mathcal{C}_k}}$,
        the matrix $P=E^\top \tilde{P}E=\Sigma_{k=1}^q E^\top_{\mathcal{C}_k} \tilde{P}_k E_{\mathcal{C}_k}$ is positive definite and belongs to $\mathcal{S}$.
\end{lem}

The next assumption on graph structure is required for the block-diagonal factorization of sparse matrices.
This assumption means that all the nodes and edges are covered with cliques in $\mathcal{Q}_\mathcal{G}$ and can hold by appropriate choice of $\mathcal{Q}_\mathcal{G}$ for any graph.
\begin{assumption}\label{assumption graph}
        For an undirected communication graph $\mathcal{G}$ and the clique index set $\mathcal{Q}_\mathcal{G}=\{1,\dots,q\}$, $\mathcal{Q}^i_\mathcal{G}$ satisfies
        \begin{itemize}
                \item $\mathcal{Q}_{\mathcal{G}}^i \neq \emptyset$ for all $i \in \mathcal{N}$.
                \item $\mathcal{Q}_{\mathcal{G}}^i \cap \mathcal{Q}_{\mathcal{G}}^j \neq \emptyset \iff (i, j) \in \mathcal{E}$.
        \end{itemize}
\end{assumption}

The next lemma shows that for all undirected graphs, every matrix $K\in\mathcal{S}$ has a block-diagonal factorization.
\begin{lem}[\cite{watanabe2024convex}] \label{lem: StoBlkdiag}
        Suppose Assumption \ref{assumption graph} holds. Then, the following transformation exists.
        \begin{equation*}
                \begin{aligned}
                        \mathcal{S}=\{E^\top \tilde{G}E:&\tilde{G}={\rm blkdiag}(\dots,\tilde{G}_k,\dots),\\
                        & \ \tilde{G}_k\in\mathbb{R}^{n_{\mathcal{C}_k}\times n_{\mathcal{C}_k}},k\in\mathcal{Q}_\mathcal{G}\}
                \end{aligned}
        \end{equation*}
\end{lem}

Also, the following lemma plays a significant role in transformations of matrix inequalities:
\begin{lem}[Finsler's lemma \cite{de2007stability}]
        Let $x\in \mathbb{R}^n$, $Q=Q^\top \in \mathbb{R}^{n\times n}$, and $M\in \mathbb{R}^{r\times n}$
        such that ${\rm rank}(M)<n$. The following statements are equivalent:
        \begin{itemize}
                \item[1)] $x^\top Qx <0, \ \forall x\neq 0\ {\rm s.t.}\ Mx = 0$.
                \item[2)] $M^{\bot \top}QM^{\bot}\prec O$.
                \item[3)] $\exists \rho \in \mathbb{R} \ {\rm s.t.} \ Q+\rho M^\top M \prec O$.
                \item[4)] $\exists X\in\mathbb{R}^{n\times r} \ {\rm s.t.}\ Q+M^\top X^\top +XM \prec O$. \\
        \end{itemize}
\end{lem}

Next, we provide the explicit formulation of $\mathcal{K}_\mathcal{S}$, the solution set for the discrete-time version of the clique-wise decomposition method \cite{watanabe2024convex}.
This formulation constrains the Lyapunov matrix to satisfy a sparsity pattern, yet it yields a larger solution set than $\mathcal{K}_{\rm diag}$ in \eqref{eq: Kdiag}.
Note that there is no inclusion relationship between this set and $\mathcal{K}_{\rm ext}$ in \eqref{eq: Kext}.
\begin{lem}\label{lem: 1 linear}
        Assume Assumption \ref{assumption graph} holds. Let
        \begin{align}
                \mathcal{K}_\mathcal{S} =\{&K=(E^\top E)^{-1}E^\top \tilde{Z}\tilde{Q}^{-1}E\notag\\
                &=(E^\top E)^{-1}\sum_{k=1}^{q}E_{\mathcal{C}_k}^\top (\tilde{Z}_k\tilde{Q}_k^{-1}) E_{\mathcal{C}_k}\in\mathcal{S}:\notag\\
                &\exists \tilde{Z}={\rm{blkdiag}}(\tilde{Z}_1,\cdots\tilde{Z}_q),\notag\\
                &\tilde{Q}={\rm{blkdiag}}(\tilde{Q}_1,\cdots\tilde{Q}_q)\succ O,
                \rho\in\mathbb{R}, \ \eta > 0,\notag\\
                &{\rm{s.t.}} \ 
                \begin{bmatrix}
                        \tilde{Q}&\ast\\
                        \tilde{A}\tilde{Q}+\tilde{B}\tilde{Z}&\tilde{Q}
                \end{bmatrix}
                +\rho W\succ O,\notag \\ 
                & \hspace{1.6em} U\tilde{Q}+\tilde{Q}U \succeq \eta U\},\label{eq: thm2}
        \end{align} 
        where
        \begin{align}
            W=&{\rm blkdiag}(U,U), \ U = I-E^\top(E^\top E)^{-1}E \label{FinsM},\\
            \tilde{A} =& E A (E^\top E)^{-1} E^\top, \ \tilde{B} = E B (E^\top E)^{-1} E^\top.
        \end{align} Then $\mathcal{K}_{\rm diag} \subset \mathcal{K}_\mathcal{S} \subset \mathcal{K}_{\rm all}$ holds.
\end{lem}
\begin{proof}
        The proof follows a similar procedure as in \cite{watanabe2024convex}; see Appendix~\ref{App: A}. for the full proof.
\end{proof}

Here, the Lyapunov matrix $P = E^\top \tilde{Q}^{-1} E$ is structurally constrained in sparsity, $P \in \mathcal{S}$, as established in Lemma \ref{lem: Agler}. 
In our main result, we aim to remove structural constraints on $\tilde{Q}$, and consequently on the Lyapunov matrix $P$.

\section{MAIN RESULT}\label{sec: main result}
In this section, we present solutions to Problem \ref{prob 1}.
Consider system \eqref{eq: sys} with distributed state feedback \eqref{eq: FB}.
In the following, for the communication graph $\mathcal{G}$, we suppose that Assumption \ref{assumption graph} holds.
First, we present $\hat{\mathcal{K}}_{\mathcal{S}, {\rm ext}}$, a solution set to Problem \ref{prob 1}, characterized by non-convex optimization.
In this formulation, the block diagonal restriction is imposed on the auxiliary matrix $\tilde{G}$ rather than on $\tilde{Q}$, thereby removing the structural constraint on the Lyapunov matrix.
\begin{lem}\label{thm: ext nonlinear}
        Let
        \begin{align}
                \hat{\mathcal{K}}_{\mathcal{S}, {\rm ext}}=\{&K=(E^\top E)^{-1}E^\top \tilde{Z}\tilde{G}^{-1}E\in\mathcal{S}: \ \exists \tilde{Q}\succ O\notag\\
                 &\tilde{Z}={\rm{blkdiag}}(\tilde{Z}_1,\cdots\tilde{Z}_q),\ \rho\in\mathbb{R}\notag\\
                 &\tilde{G}={\rm{blkdiag}}(\tilde{G}_1,\cdots\tilde{G}_q)\notag\\
                 &{\rm{s.t.}} \ 
                 \Psi (\tilde{G}, \tilde{Q}, \tilde{Z})
                 +\rho \tilde{H}^\top W\tilde{H}\succ O\}\label{eq: ext nonlinear}
        \end{align}
        with $W$ in \eqref{FinsM} and
        \begin{align}
            \Psi (\tilde{G}, \tilde{Q}, \tilde{Z})&=
            \begin{bmatrix}
              \tilde{G}+\tilde{G}^\top - \tilde{Q}&\ast\\
              \tilde{A}\tilde{G}+\tilde{B}\tilde{Z}&\tilde{Q}\\
            \end{bmatrix}, \label{psi}\\
            \tilde{H} \ \ \ &={\rm blkdiag}(\tilde{G},\tilde{G}).\label{eq: H}
        \end{align}
        Then, $\mathcal{K}_\mathcal{S}\cup \mathcal{K}_{\rm ext}\subset\hat{\mathcal{K}}_{\mathcal{S}, {\rm ext}}\subset \mathcal{K}_{\rm all}$ holds.
\end{lem}
\begin{proof}
        First, we prove that $\hat{\mathcal{K}}_{\mathcal{S}, {\rm ext}} \subset \mathcal{K}_{\rm all}$.
        Consider $K = (E^\top E)^{-1} E^\top \tilde{Z} \tilde{G}^{-1} E \in \hat{\mathcal{K}}_{\mathcal{S}, {\rm ext}}$.
        Here, $K$ can be constructed since $\tilde{G}$ is nonsingular, as shown below.
       From the block-diagonal elements of the matrix inequality in \eqref{eq: ext nonlinear}, we obtain 
       $\tilde{G}+\tilde{G}^\top-\tilde{Q}+\rho\tilde{G}^{\top}U\tilde{G}\succ O$ and 
       $\tilde{Q}+\rho\tilde{G}^{\top}U\tilde{G}\succ O$. 
       Combining these matrix inequalities yields
        \begin{align}
                O &\prec \tilde{G}+\tilde{G}^\top+2\rho\tilde{G}^{\top}U\tilde{G}\notag\\
                & = (I+\rho\tilde{G}^{\top}U)\tilde{G}+\tilde{G}^\top (I+\rho\tilde{G}^{\top}U)^\top.\label{eq: tildeG nonsingular}
        \end{align}
       \eqref{eq: tildeG nonsingular} shows the nonsingularity of $\tilde{G}$ and $\tilde{H}$ in \eqref{eq: ext nonlinear}; otherwise, for a non-zero vector $x\in \ker{\tilde{G}}$, pre- and post-multiplying \eqref{eq: tildeG nonsingular} by $x^\top$ and $x$, respectively, would yield zero.

       Then, by pre- and post-multiplying the matrix inequality in \eqref{eq: ext nonlinear} with $\tilde{H}^{-\top}$ and $\tilde{H}^{-1}$, we get
       \begin{align}
        O\prec&\ \tilde{H}^{-\top}\Psi(\tilde{G}, \tilde{Q}, \tilde{Z})\tilde{H}^{-1} +\rho W\notag\\
        =&\ 
        \begin{bmatrix}
                \tilde{G}^{-\top}+\tilde{G}^{-1} - \tilde{G}^{-\top}\tilde{Q}\tilde{G}^{-1}&\hspace{-1em}\ast\\
                \tilde{G}^{-\top}(\tilde{A}+\tilde{B}\tilde{K})&\hspace{-1.5em}\tilde{G}^{-\top}\tilde{Q}\tilde{G}^{-1}\\
        \end{bmatrix}
        +\rho W,\label{eq: Lemma 6 eq 15}
       \end{align}
       where $\tilde{K}=\tilde{Z}\tilde{G}^{-1}$.
       We can transform \eqref{eq: Lemma 6 eq 15} by using Finsler's lemma and setting $X= E^\top\tilde{G}^{-1}E, P=E^\top\tilde{G}^{-\top}\tilde{Q}\tilde{G}^{-1}E$.
       Here, $X$ is nonsingular, because $\tilde{G}$ is nonsingular and $E$ has full column rank.
       \begin{align}
        &
        \begin{bmatrix}
                E^\top \ O \\
                O \ E^\top
        \end{bmatrix}
        \begin{bmatrix}
                \tilde{G}^{-\top}+\tilde{G}^{-1} - \tilde{G}^{-\top}\tilde{Q}\tilde{G}^{-1}&\hspace{-1em}\ast\\
                \tilde{G}^{-\top}(\tilde{A}+\tilde{B}\tilde{K})&\hspace{-1.5em}\tilde{G}^{-\top}\tilde{Q}\tilde{G}^{-1}\\
        \end{bmatrix}
        \begin{bmatrix}
                E \ O \\
                O \ E
        \end{bmatrix}\notag\\
        = &
        \begin{bmatrix}
                X^\top +X-P&\ast\\
               X^\top(A+BK)&P\\
        \end{bmatrix} \
        \succ\ O. \label{eq: ext1}
       \end{align}
       Note that $W^\bot = {\rm blkdiag}(E^\top, E^\top)$ is used for the Finsler's lemma. 
       Here, \eqref{eq: ext1} shares a similar structure with the extended LMI \cite{de2002extended} and can likewise be transformed by utilizing the following inequality,
       \begin{align}
        &(P-X^\top)^\top P^{-1} (P-X^\top) \succeq O\notag\\
        \Leftrightarrow\ &
        XP^{-1}X^\top \succeq X+X^\top -P,
       \end{align}
       allowing a transformation of \eqref{eq: ext1} as follows:
       \begin{align*}
        &
        \begin{bmatrix}
                XP^{-1}X^\top&\ast\\
               X^\top(A+BK)&P\\
        \end{bmatrix}\succ O \notag\\
        \Leftrightarrow\ &
        \begin{bmatrix}
                XP^{-1}X^\top&\ast\\
               A+BK&X^{-\top} PX^{-1}\\
        \end{bmatrix}\succ O \notag\\
        \Leftrightarrow\ &
        (A+BK)^\top XP^{-1}X^\top (A+BK)- XP^{-1}X^\top \prec O.
       \end{align*}
       Since $P^{-1}\succ O$,
       $XP^{-1}X^\top$
       is a positive definite Lyapunov matrix. Therefore, $K\in \mathcal{K}_{\rm all}$ and $\hat{\mathcal{K}}_{\mathcal{S}, {\rm ext}}\subset \mathcal{K}_{\rm all}$ hold.

       The proof of $\mathcal{K}_\mathcal{S}\cup \mathcal{K}_{\rm ext}\subset\hat{\mathcal{K}}_{\mathcal{S}, {\rm ext}}$ is omitted, as it follows directly from the next theorem.
\end{proof}

Next, the following Theorem gives our main result, convexly characterized solution to Problem \ref{prob 1}.
This theorem is derived as a convex relaxation of Lemma \ref{thm: ext nonlinear}.
\begin{thm}\label{thm: ext linear}
        Let
        \begin{align}
                \mathcal{K}_{\mathcal{S}, {\rm ext}}=\{&K=(E^\top E)^{-1}E^\top \tilde{Z}\tilde{G}^{-1}E\in\mathcal{S}: \ \exists \tilde{Q}\succ O\notag\\
                 &\tilde{Z}={\rm{blkdiag}}(\tilde{Z}_1,\cdots\tilde{Z}_q),\notag\\
                 &\tilde{G}={\rm{blkdiag}}(\tilde{G}_1,\cdots\tilde{G}_q), \notag \\ 
                 &\rho\in\mathbb{R}, \ \eta >0 \ {\rm{s.t.}}\notag\\ 
                 &\Psi (\tilde{G}, \tilde{Q}, \tilde{Z})
                 +\rho W\succ O,\notag \\
                 &\tilde{G}^\top U+U\tilde{G}\succeq \eta U\} \label{eq: stb ext rlx}
        \end{align}
        with $W, U$ in \eqref{FinsM} and $\Psi$ in \eqref{psi}.
        Then, $\mathcal{K}_\mathcal{S}\cup \mathcal{K}_{\rm ext}\subset\mathcal{K}_{\mathcal{S}, {\rm ext}}\subset \mathcal{K}_{\rm all}$ holds.
\end{thm}
\begin{proof}
        First, we prove $\mathcal{K}_{\mathcal{S}, {\rm ext}}\subset \hat{\mathcal{K}}_{\mathcal{S}, {\rm ext}}$ to show $\mathcal{K}_{\mathcal{S}, {\rm ext}}\subset \mathcal{K}_{\rm all}$.
        Consider $K=(E^\top E)^{-1}E^\top \tilde{Z}\tilde{G}^{-1}E\in \mathcal{K}_{\mathcal{S}, {\rm ext}}$.
        We first demonstrate the nonsingularity of $\tilde{G}$ to show that $K$ can be constructed. From the upper-left block-diagonal element of $\Psi + \rho W \succ O$ in \eqref{eq: stb ext rlx}, we obtain
        \begin{align*}
                O\prec \tilde{Q}&\prec\tilde{G}+\tilde{G}^\top +\rho U\\
                &\preceq \tilde{G}+\tilde{G}^\top + {|\rho|}/{\eta}(\tilde{G}^\top U + U\tilde{G})\\
                &= ({|\rho|}/{\eta}\ U+I)\tilde{G}+\tilde{G}^\top({|\rho|}/{\eta}\ U+I),
        \end{align*}
        which proves the nonsingularity of $\tilde{G}$, and thus of $\tilde{H}$ in \eqref{eq: H}.
        Then from the LMI conditions in \eqref{eq: stb ext rlx}, 
        \begin{align*}
                O\prec\ & \Psi (\tilde{G}, \tilde{Q}, \tilde{Z})+|\rho| W\\
                \preceq\ & \Psi (\tilde{G}, \tilde{Q}, \tilde{Z})+|\rho|/\eta\  (\tilde{H}^\top W+W\tilde{H})\\
                \Leftrightarrow O\prec\ &
                \tilde{H}^{-\top}\Psi (\tilde{G}, \tilde{Q}, \tilde{Z})\tilde{H}^{-1}+|\rho|/\eta\  (W\tilde{H}^{-1}+\tilde{H}^{-\top} W)
        \end{align*}
        is achieved.
        By applying transformations from statements 4) to 3) of Finsler's lemma, we obtain
        \begin{equation}
            O\prec \tilde{H}^{-\top}\Psi (\tilde{G}, \tilde{Q}, \tilde{Z})\tilde{H}^{-1} + \mu W^{\top}W\label{eq: H Psi H}
        \end{equation}
        for some $\mu \in \mathbb{R}$.
        Since $W^{\top}W=W$, pre- and post-multiplying \eqref{eq: H Psi H} with $\tilde{H}^{\top}$ and $\tilde{H}$ yields the matrix inequality for $\hat{\mathcal{K}}_{\mathcal{S}, {\rm ext}}$ in \eqref{eq: ext nonlinear}, implying that $\mathcal{K}_{\mathcal{S}, {\rm ext}} \subset \hat{\mathcal{K}}_{\mathcal{S}, {\rm ext}}$.
    
        Next, for $\mathcal{K}_\mathcal{S}\subset \mathcal{K}_{\mathcal{S}, {\rm ext}}$, one can choose $\tilde{G}=\tilde{G}^\top=\tilde{Q}$ to recover the LMIs in \eqref{eq: thm2}.
    
        Lastly, we prove $\mathcal{K}_{\rm ext}\subset \mathcal{K}_{\mathcal{S}, {\rm ext}}$.
        For $K\in \mathcal{K}_{\rm ext}$, we have $K=ZG^{-1}$ with some $Z\in\mathcal{S}$ and $G={\rm blkdiag}(G_1,\dots,G_N)$.
        From Lemma \ref{lem: StoBlkdiag}, a block-diagonal matrix $\tilde{K}$ exists such that $K = (E^\top E)^{-1} E^\top \tilde{K} E$.
        We define $\tilde{G} = \text{blkdiag}(\cdots, \tilde{G}_k, \cdots)$, where $\tilde{G}_k = \text{blkdiag}(\cdots, |\mathcal{Q}_\mathcal{G}^j| G_j, \cdots)$ and $j \in \mathcal{C}_k$.        
        Then, $\tilde{G}$ is nonsingular because $|\mathcal{Q}_\mathcal{G}^j|\neq 0$ and $G+G^\top \succ O$. 
        By setting $\tilde{Z} = \tilde{K} \tilde{G}$, we obtain $(E^\top E)^{-1} E^\top \tilde{Z} \tilde{G}^{-1} E = (E^\top E)^{-1} E^\top \tilde{Z} E (E^\top E)^{-1} G^{-1} = Z G^{-1} = K$,
        since $Z$ can be represented as $Z = (E^\top E)^{-1} E^\top \tilde{Z} \tilde{G}^{-1} E \in \mathcal{S}$ from Lemma \ref{lem: StoBlkdiag}, and $\tilde{G}^{-1} E = E (E^\top E)^{-1} G^{-1}$.
        Since $E$ has full column rank, for any positive definite matrix $P$, there exists a positive definite matrix $\tilde{P}$ such that $P = E^\top \tilde{P} E$.
        Therefore, for $G^{-\top}QG^{-1}\succ O$, there exists $\tilde{Q}\succ O$ such that $G^{-\top}QG^{-1}=E^\top\tilde{G}^{-\top}\tilde{Q}\tilde{G}^{-1}E$. 
        Using these transformations, the LMI in \eqref{eq: Kext} can be reformulated as follows:
        \begin{align*}
                O\prec&
                \begin{bmatrix}
                        G^{-\top}+G^{-1}-G^{-\top}QG^{-1}&\ast\\
                        G^{-\top}(A+BK)&G^{-\top}QG^{-1}
                \end{bmatrix}\\
                =&
                \begin{bmatrix}
                        \hspace{-3em}E^\top (\tilde{G}^{-\top}+ \tilde{G}^{-1}-\tilde{G}^{-\top}\tilde{Q}\tilde{G}^{-1})E\hspace{3em} \ast\hspace{3em}\\
                        \hspace{1em}E^\top \tilde{G}^{-\top}E(A+BK)\hspace{4em} E^\top\tilde{G}^{-\top}\tilde{Q}\tilde{G}^{-1}E
                \end{bmatrix}\\
                =&
                \begin{bmatrix}
                        E^\top O \\
                        O \ E^\top
                \end{bmatrix}
                \begin{bmatrix}
                        \tilde{G}^{-\top}+\tilde{G}^{-1} - \tilde{G}^{-\top}\tilde{Q}\tilde{G}^{-1}&\hspace{-1.5em}\ast\\
                        \hspace{-1em}\tilde{G}^{-\top}(\tilde{A}+\tilde{B}\tilde{K})&\hspace{-2.5em}\tilde{G}^{-\top}\tilde{Q}\tilde{G}^{-1}\\
                \end{bmatrix}
                \begin{bmatrix}
                        E \ O \\
                        O \ E
                \end{bmatrix}
        \end{align*}
        By using Finsler's lemma, we obtain
        \begin{align}
                O&\prec
                \begin{bmatrix}
                        \hspace{-2em}\tilde{G}^{-\top}+ \tilde{G}^{-1}-\tilde{G}^{-\top}\tilde{Q}\tilde{G}^{-1}\hspace{1em}\ast\hspace{1em}\\
                        \tilde{G}^{-\top}(\tilde{A}+\tilde{B}\tilde{K})\hspace{3em}\tilde{G}^{-\top}\tilde{Q}\tilde{G}^{-1}
                \end{bmatrix}
                +\nu W \label{eq: ext diag proof}
        \end{align}
        for some $\nu \in\mathbb{R}$. Since $\tilde{G}U=U\tilde{G} \Leftrightarrow \tilde{H}W=W\tilde{H}$ and $W=W^2$, we obtain
        \begin{align*}
                \eqref{eq: ext diag proof} \Leftrightarrow O&\prec 
                \Psi (\tilde{G}, \tilde{Q}, \tilde{Z})+\nu \tilde{H}^\top W^2\tilde{H}\notag \\
                &=
                \Psi (\tilde{G}, \tilde{Q}, \tilde{Z})+\nu W\tilde{H}^\top \tilde{H}W\notag\\
                &\preceq
                \Psi (\tilde{G}, \tilde{Q}, \tilde{Z})+\rho W,\notag
        \end{align*}
        for $\rho={\rm max}(\nu \ {\rm eig}(\tilde{H}^\top \tilde{H}))$.
        Moreover, $\tilde{G}^\top U+U\tilde{G}=\tilde{G}^{\top}U^2+U^2\tilde{G}=U(\tilde{G}^\top+\tilde{G})U\succeq \mu U$ holds, where $\mu$ is the minimum or maximum eigenvalue of $\tilde{G}^\top+\tilde{G}$.
        This satisfies the LMI in \eqref{eq: stb ext rlx}.
        Therefore, we obtain $K\in \mathcal{K}_{\mathcal{S}, {\rm ext}}$ and $\mathcal{K}_{\rm ext}\subset \mathcal{K}_{\mathcal{S}, {\rm ext}}$.
\end{proof}

\begin{rem}
        For Theorem 1, any arbitrary choice of \( \mathcal{Q}_\mathcal{G} \) that satisfies Assumption 1 will yield a larger solution set compared to the Extended LMI in (6).
        It is advisable to set $\mathcal{Q}_\mathcal{G} = \mathcal{Q}_\mathcal{G}^{\rm max}$, so that the dilated matrices $\tilde{G}$ in Theorem 1 have a structure closer to a full matrix. 
        Note that in distributed systems, each agent can compute the maximal cliques it belongs to with a computation time of $\mathcal{O}(3^{|\mathcal{N}_i|/3})$, where $|\mathcal{N}_i|$ denotes the number of neighboring agents \cite{sakurama2021generalized}.
        This is not a heavy computation burden for a sparse graph.
\end{rem}
\begin{rem}
        When $\mathcal{Q}_\mathcal{G}$ is chosen appropriately, the orders of the number of SDP variables are  $\mathcal{O}(nd)$, $\mathcal{O}(nd)$, $\mathcal{O}(n^2)$ and $\mathcal{O}(n^2d^2)$ for $\mathcal{K}_{\rm diag}$,$\mathcal{K}_\mathcal{S}$, $\mathcal{K}_{\rm ext}$ and $\mathcal{K}_{\mathcal{S}, {\rm ext}}$, respectively, where $d$ denotes the maximum number of neighboring agents.
        Here, the number of SDP variables of the proposed method $\mathcal{K}_{\mathcal{S}, {\rm ext}}$, $\mathcal{O}(n^2d^2)$, is the largest.
        However, for a large, sparse graph satisfying $n\gg d$, the number approximates $\mathcal{O}(n^2)$, the same as the number of SDP variables of the extended LMIs, $\mathcal{K}_{\rm ext}$.
        See Appendix~\ref{App: B}. for details.
\end{rem}


\section{Application to $H_\infty$ control}\label{sec: Hinf}
In this section, we present a solution to the $H_\infty$ version of Problem \ref{prob 1}.
Consider the following system, where $w\in\mathbb{R}^{m_w}$ and $y\in\mathbb{R}^l$ are the exogenous disturbance and the output, $B_w\in\mathbb{R}^{n\times m_v}$, $D_w\in\mathbb{R}^{l\times m_w}$, $C\in\mathbb{R}^{l\times n}$, and $D\in\mathbb{R}^{l\times m}$ are the disturbance, output and feedthrough matrices, respectively:
\begin{equation}\label{eq: Hsys}
        \begin{aligned}
                x(k+1)&=Ax(k)+Bu(k)+B_ww(k)\\
                y(k)&=Cx(k)+Du(k)+D_ww(k).
        \end{aligned}
\end{equation}
We consider the distributed \( H_\infty \) state-feedback control for the system \eqref{eq: Hsys}.
The set of all distributed control gains achieving the system's $H_\infty$ norm less than $\gamma$ is defined as:
\begin{equation}
        \mathcal{K}_{\rm all}^{\infty, \gamma}=\{K\in\mathcal{S}: ||T(z)||_\infty < \gamma\},
\end{equation}
where $T(z)=C(zI - (A+BK))^{-1} B_w + D_w$. 
Here, $H_\infty$ controller synthesis can be viewed as a minimization problem with the objective function $\gamma$ and the variable $K\in\mathcal{S}$. 
Similarly to Problem 1, we expect to find a convexly relaxed subset of $\mathcal{K}_{\rm all}^{\infty, \gamma}$ that is less conservative than the solution sets obtained by clique-wise decomposition or the extended LMI.

The following lemma characterizes the $H_\infty$ norm of the system \eqref{eq: Hsys}.
Here, statement 3) is the extended LMI version of the Bounded Real Lemma \cite{de2002extended}.
\begin{lem}[Bounded Real Lemma \cite{gahinet1994,de1999lmi,de2002extended}]\label{lem: BRL}
        The following statements are equivalent:
        \begin{itemize}
                \item[1)] $A+BK$ is Schur and $||T(z)||_\infty < \gamma$.
                \item[2)] $\exists P\succ O,\ {\rm s.t.}$
                \begin{equation}
                        \begin{bmatrix}
                                -P^{-1}&\ast&\ast&\ast\\
                                (A+BK)^\top&-P&\ast&\ast\\
                                B_w^\top&O&-\gamma I&\ast\\
                                O&C&D_w&-\gamma I\\
                        \end{bmatrix}
                        \prec O \label{BRL}
                \end{equation}
                \item[3)] $\exists P\succ O, G\ {\rm s.t.}$
                \begin{equation}
                        \begin{bmatrix}
                                -P&\ast&\ast&\ast\\
                                G^\top (A+BK)^\top&P-G-G^\top&\ast&\ast\\
                                B_w^\top&O&-\gamma I&\ast\\
                                O&CG&D_w&-\gamma I\\
                        \end{bmatrix}
                        \prec O \\ \label{Extended BRL}
                \end{equation}
        \end{itemize}
\end{lem}
Using Lemma \ref{lem: BRL}, 2), $\mathcal{K}_{\rm all}^{\infty, \gamma}$ can be transformed to:
\begin{align*}
    \mathcal{K}_{\rm all}^{\infty, \gamma}=\{&K \in \mathcal{S}:\exists P\succ O,{\rm s.t.}\  \eqref{BRL}\}.
\end{align*}
Similarly to $\mathcal{K}_{\rm diag}$, the block-diagonal convex relaxation of $\mathcal{K}_{\rm all}^{\infty, \gamma}$ , denoted as $\mathcal{K}_{\rm diag}^{\infty, \gamma}$, is given as follows:
\begin{align}
    \mathcal{K}_{\rm diag}^{\infty, \gamma}=\{&K=ZQ^{-1}:\exists Z\in \mathcal{S},\notag \\
      &\ Q={\rm blkdiag}(Q_1,\cdots,Q_N)\succ O \ {\rm s.t.}\notag\\
     &
      \begin{bmatrix}
        \hspace{0em}-Q \hspace{3.5em} \ast&\ast&\ast\\
        \hspace{-1em}(AQ+BZ)^\top \hspace{0.5em} -Q&\ast&\ast\\
        \hspace{0em}B_w^\top \hspace{3.5em} O &-\gamma I&\ast\\
        \hspace{2em}O \hspace{2.5em} CQ+DZ&D_w&-\gamma I\\
      \end{bmatrix}
    \prec O\}.\label{eq: Hinf diag}
\end{align}

Similarly to $\mathcal{K}_{\rm ext}$ and $\mathcal{K}_\mathcal{S}$, $H_\infty$ versions of the Extended LMI method and the clique-wise decomposition method can be derived as follows.
\begin{lem}[Distributed version of \cite{de1999lmi, de2002extended}]\label{lem: extended lmi Hinf}
        Let
        \begin{align}
        \mathcal{K}_{\rm ext}^{\infty, \gamma}=\{&K=ZG^{-1}:, \ \exists Z\in \mathcal{S}, \ Q \succ O,\notag\\\notag
        &G ={\rm blkdiag}(G_1,\cdots,G_N) \ \ {\rm s.t.}\\\notag
        &
        \begin{bmatrix}
                -Q&\ast&\ast&\ast\\
                (AG+BZ)^\top&Q-G-G^\top&\ast&\ast\\
                B_w^\top&O&-\gamma I&\ast\\
                O&CG+DZ&D_w&-\gamma I\\
        \end{bmatrix}\\
        &\prec O
        \}.\label{K_ext,diag}
        \end{align}\label{eq: Hinf ext}
        Then, $\mathcal{K}_{\rm diag}^{\infty, \gamma}\subset\mathcal{K}_{\rm ext}^{\infty, \gamma}\subset\mathcal{K}_{\rm all}^{\infty, \gamma}$ holds.
\end{lem}
\begin{proof}
        This can be proved by following the similar procedure to Lemma \ref{lem: extended lmi}.
\end{proof}

\begin{lem}\label{thm: Hinf linear} 
        Let
        \begin{align}
          \mathcal{K}_\mathcal{S}^{\infty, \gamma}=\{&K=(E^\top E)^{-1}E^\top \tilde{Z}\tilde{Q}^{-1}E:\notag\\
          &\exists \tilde{Z}={\rm blkdiag}(\tilde{Z}_1,\cdots\tilde{Z}_q),\notag\\
          &\tilde{Q}={\rm blkdiag}(\tilde{Q}_1,\cdots,\tilde{Q}_q)\succ O, \notag\\
          &\rho \in \mathbb{R}, \eta > 0 \ {\rm s.t.}\notag\\
          &\Gamma_\gamma(\tilde{Q},\tilde{Z})+{\rm blkdiag}(\rho W, O_{2n\times 2n})\prec O,\notag\\
          & U\tilde{Q}+\tilde{Q}U \succeq \eta U\}\label{eq: Hinf Ks}
        \end{align}
        with $W$, $U$ in \eqref{FinsM},
        \begin{align}
                &\Gamma_\gamma(\tilde{Q},\tilde{Z})=
                \begin{bmatrix}
                  -\tilde{Q}&\ast&\ast&\ast\\
                  (\tilde{A}\tilde{Q}+\tilde{B}\tilde{Z})^\top&-\tilde{Q}&\ast&\ast\\
                  \tilde{B}_v^\top&O&-\gamma I&\ast\\
                  O&\tilde{C}\tilde{Q}+\tilde{D}\tilde{Z}&D_w&-\gamma I
                \end{bmatrix}\label{eq: gamma},
        \end{align}
        and dilated system matrices
        \begin{align}
                &\tilde{A} = E A (E^\top E)^{-1} E^\top, \ \tilde{B} = E B (E^\top E)^{-1} E^\top,\notag\\
                &\tilde{C} = C (E^\top E)^{-1} E^\top, \tilde{D} = D (E^\top E)^{-1} E^\top, \tilde{B}_v = E B_w.\notag
        \end{align}
        Then, $\mathcal{K}_{\rm diag}^{\infty,\gamma}\subset\mathcal{K}_\mathcal{S}^{\infty, \gamma} \subset \mathcal{K}_{\rm all}^{\infty, \gamma}$ holds.
\end{lem}
\begin{proof}
        This can be proved by following the similar procedure to Lemma \ref{lem: 1 linear} and \cite{watanabe2024convex}.
\end{proof}

The following theorem provides a less conservative convex solution set to the $H_\infty$ control version of Problem \ref{prob 1} than both $\mathcal{K}_{\rm ext}^{\infty, \gamma}$ and $\mathcal{K}_\mathcal{S}^{\infty, \gamma}$.
A clique-wise decomposition is applied to the auxiliary matrix $G$ in \eqref{eq: Hinf ext}.
\begin{thm}\label{thm: Hinf ext linear}
        Let
        \begin{align}
          \mathcal{K}_{\mathcal{S}, {\rm ext}}^{\infty, \gamma}=\{&K=(E^\top E)^{-1}E^\top \tilde{Z}\tilde{G}^{-1}E:\ \exists \tilde{Q}\succ O,\notag\\
          &\tilde{Z}={\rm blkdiag}(\tilde{Z}_1,\cdots\tilde{Z}_q),\notag\\
          &\tilde{G}={\rm blkdiag}(\tilde{G}_1,\cdots\tilde{G}_q), \notag\\
          &\rho \in \mathbb{R}, \ \eta >0, \ {\rm s.t.}\notag\\
          &\Theta_\gamma(\tilde{G},\tilde{Q},\tilde{Z})+{\rm blkdiag}(\rho W, O_{2n\times 2n})\prec O,\notag\\
          &\tilde{G}^\top U+U\tilde{G}\succeq \eta U\}\label{eq: Hinf Ksext}
        \end{align}
        with $W$, $U$ in \eqref{FinsM}, $\tilde{H}$ in \eqref{eq: H} and 
        \begin{align}
            \Theta_\gamma(\tilde{G},\tilde{Q},\tilde{Z})&=\notag\\
            & \hspace{-2.5em} 
            \begin{bmatrix}
              -\tilde{Q}&\ast&\ast&\ast\\
              (\tilde{A}\tilde{G}+\tilde{B}\tilde{Z})^\top&\tilde{Q}-\tilde{G}-\tilde{G}^\top&\ast&\ast\\
              \tilde{B}_v^\top&O&-\gamma I&\ast\\
              O&\tilde{C}\tilde{G}+\tilde{D}\tilde{Z}&D_w&-\gamma I
            \end{bmatrix}\label{Theta}.
        \end{align}
        Then, $\mathcal{K}_{\rm ext}^{\infty,\gamma}\cup\mathcal{K}_\mathcal{S}^{\infty, \gamma}\subset\mathcal{K}_{\mathcal{S}, {\rm ext}}^{\infty, \gamma}\subset \mathcal{K}_{\rm all}^{\infty, \gamma}$ holds.
\end{thm}
\begin{proof}
    See Appendix~\ref{App: C}. for the proof details.
\end{proof}

\section{NUMERICAL EXAMPLES}\label{sec: numerical examples}
Here, we compare the numerical performances in $H_\infty$ control presented in section~\ref{sec: Hinf} 
with the performance of centralized $H_\infty$ control, which considers a complete graph $\mathcal{G}$.
We computed the minimum value of $\gamma$, the upper bound of the $H_\infty$ norm $T(z)$ for the system \eqref{eq: Hsys} with a distributed controller $K \in \mathcal{S}$, using \eqref{eq: Hinf diag}, \eqref{K_ext,diag}, \eqref{eq: Hinf Ks}, and \eqref{eq: Hinf Ksext}, respectively.
We consider scenarios that $N$ agents are randomly partitioned into $l$ groups, and these groups are arranged sequentially, with each group connected to its adjacent groups by a single edge.
Here, we set the communication graphs within each group to be complete.
The matrix $ A $ is generated by randomly assigning poles within the range $[1, 5]$, each element of $ B $ is a random value in the range $[0, 1]$, and $C = D = D_w = B_w = I_N$.
The matrix $E$ was constructed by choosing $\mathcal{Q}_\mathcal{G}=\mathcal{Q}_\mathcal{G}^{\rm max}$.
The detailed simulation settings can be found in our GitHub repository \cite{github}.

The simulation results of 30 cases for $(N,l)=(10,3),(40,5)$ are shown in Fig.~\ref{fig: 08}.
The horizontal axis represents the sample number, and the vertical axis represents the ratio $\gamma_\ast/\gamma_{\rm cen}$, 
where $\gamma_\ast$ is the $H_\infty$ norm of each methods, and $\gamma_{\rm cen}$ is of the centralized method. 
Therefore, a smaller $\gamma_\ast/\gamma_{\rm cen}$ indicates a better result.
From Fig. 2, it can be observed that proposed method outperforms existing methods, regardless of the scale of the system.
These results demonstrate the efficacy of the proposed approach.
\begin{figure}[b]
        \centering
        \includegraphics[width = 9cm]{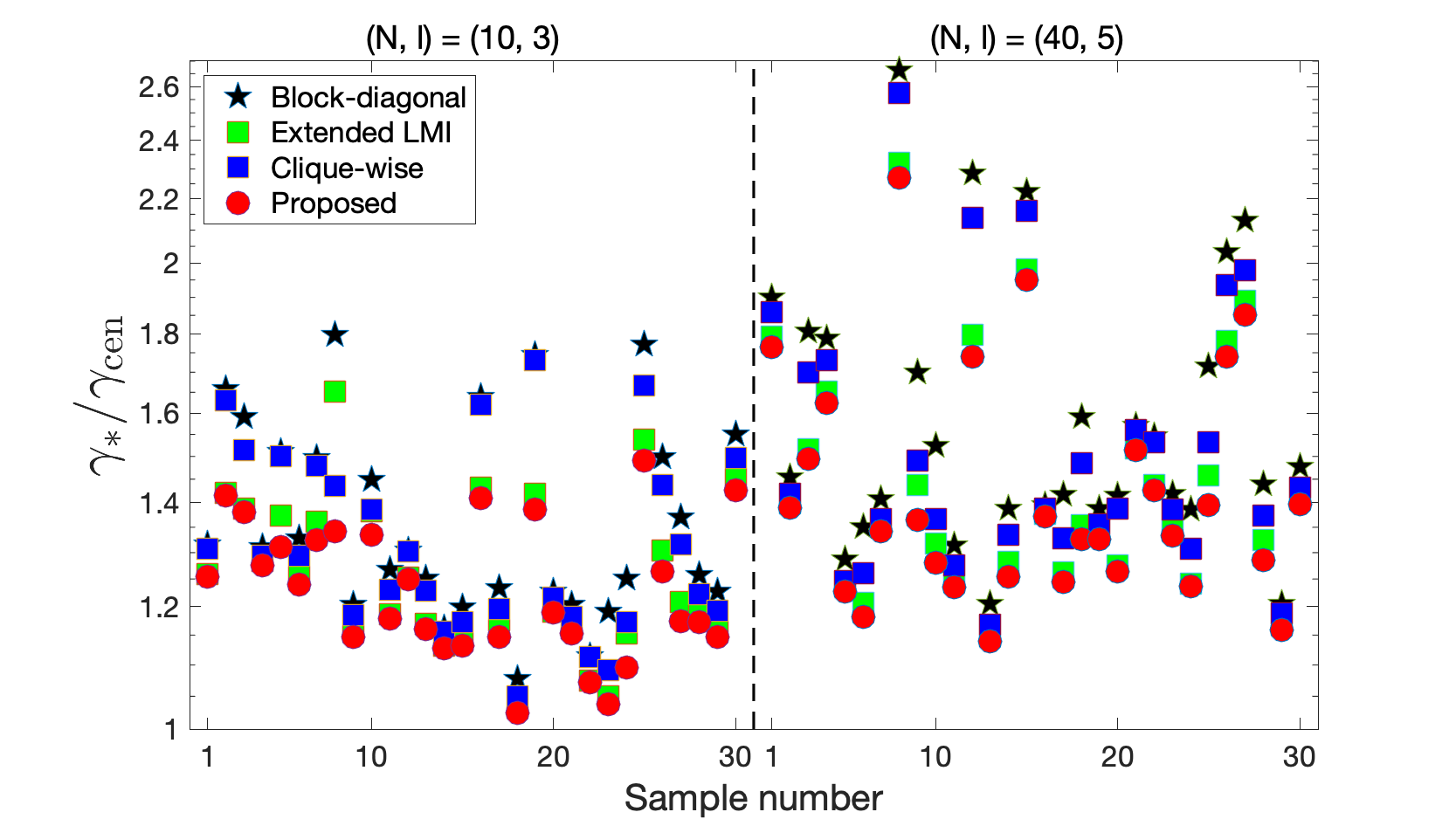}
        \caption{{\small
                Results for $H_\infty$ control.
                The black, green, blue and red markers represent the value for 
                $\mathcal{K}_{\rm diag}^{\infty,\gamma}$, 
                $\mathcal{K}_{\rm ext}^{\infty,\gamma}$,
                $\mathcal{K}_\mathcal{S}^{\infty, \gamma}$
                and the proposed method $\mathcal{K}_{\mathcal{S}, {\rm ext}}^{\infty, \gamma}$, 
                respectively.
                }}\label{fig: 08}
\end{figure}

\section{CONCLUSION}
This study addresses the distributed controller design problem for discrete-time systems. 
We derived a less conservative convex relaxation of the problem formulation by utilizing the clique-wise decomposition method and Extended LMIs.
Our future work includes providing a less conservative and numerically stable convexification of Lemma \ref{thm: ext nonlinear}, as well as investigating the efficient selection of the clique set, $\mathcal{Q}_\mathcal{G}$.


\bibliographystyle{IEEEtran}
\bibliography{main}

\appendix
\subsection{Generalized Notation}\label{App: general notation}
In the case of $n_i\neq m_i$, instead of $E$ in \eqref{eq: E}, two matrices $E_x$ and $E_u$ are defined as
\begin{equation}
  \begin{aligned}
          E_{x} &= [\dots, E_{x, \mathcal{C}_k}^\top,\dots]^\top \in \mathbb{R}^{(\Sigma_{k\in\mathcal{Q}_\mathcal{G}} n_{\mathcal{C}_k})\times n},\ k\in\mathcal{Q}_\mathcal{G},\\
          E_{x, \mathcal{C}_k}&= [\dots, E_{x, j}^\top, \dots]^\top \in \mathbb{R}^{n_{\mathcal{C}_k}\times n},\ j\in \mathcal{C}_k,\\
          E_{x, j} & = [O_{n_j\times n_1},\dots,I_{n_j},\dots,O_{n_j\times n_N}]\in\mathbb{R}^{n_j\times n},
  \end{aligned}
\end{equation}
with $n_{\mathcal{C}_k}=\Sigma_{j\in\mathcal{C}_k} n_j$ and 
\begin{equation}
  \begin{aligned}
          E_u &= [\dots, E_{u, \mathcal{C}_k}^\top,\dots]^\top \in \mathbb{R}^{(\Sigma_{k\in\mathcal{Q}_\mathcal{G}} m_{\mathcal{C}_k})\times m},\ k\in\mathcal{Q}_\mathcal{G},\\
          E_{u, \mathcal{C}_k}&= [\dots, E_{u, j}^\top, \dots]^\top \in \mathbb{R}^{m_{\mathcal{C}_k}\times m},\ j\in \mathcal{C}_k,\\
          E_{u, j} & = [O_{m_j\times m_1},\dots,I_{m_j},\dots,O_{m_j\times m_N}]\in\mathbb{R}^{m_j\times m},
  \end{aligned}
\end{equation}
with $m_{\mathcal{C}_k}=\Sigma_{j\in\mathcal{C}_k} m_j$. 
Then, Lemma \ref{lem: Agler} and \ref{lem: StoBlkdiag} are alternated as follows:\\
\textit{Lemma 2'}
Consider undirected graph $\mathcal{G}=(\mathcal{N}, \mathcal{E})$ with clique index set $\mathcal{Q}_\mathcal{G}=\{1,\dots,q\}$.
For $\tilde{P}={\rm blkdiag}(\tilde{P}_1,\dots\tilde{P}_q)\succ O$ with $\tilde{P}_k \in \mathbb{R}^{n_{\mathcal{C}_k}\times n_{\mathcal{C}_k}}$,
the matrix $P=E_x^\top \tilde{P}E_x=\Sigma_{k=1}^q E^\top_{x, \mathcal{C}_k} \tilde{P}_k E_{x, \mathcal{C}_k}$ is positive definite and belongs to $\mathcal{S}_{\rm sq}$, defined as 
\begin{equation}
  \begin{aligned}
    \mathcal{S}_{\rm sq}=\{E_x^\top \tilde{Q}E_x:\ &\tilde{Q}={\rm blkdiag}(\dots,\tilde{Q}_k,\dots),\\ &\tilde{Q}_k\in\mathbb{R}^{n_{\mathcal{C}_k}\times n_{\mathcal{C}_k}},k\in\mathcal{Q}_\mathcal{G}\}.\label{eq: Ssq}
  \end{aligned}
\end{equation}
\textit{Lemma 3'}
Suppose Assumption 1 holds. Then, the following transformation exists.
\begin{equation}
  \begin{aligned}
    \mathcal{S}=\{E_u^\top \tilde{Z}E_x:\ &\tilde{Z}={\rm blkdiag}(\dots,\tilde{Z}_k,\dots),\\ &\tilde{Z}_k\in\mathbb{R}^{m_{\mathcal{C}_k}\times n_{\mathcal{C}_k}},k\in\mathcal{Q}_\mathcal{G}\}\label{eq: S}
  \end{aligned}
\end{equation}

Here, Lemma 2' is equivalent to Lemma 2, and Lemma 3' can be proven analogously to Lemma 3. Please refer to our previous work [15].

Using Lemma 2' and Lemma 3', the definition of $\mathcal{K}_\mathcal{S}$ is altered as follows:\\
\textit{Lemma 5'} Assume Assumption 1 holds. Let
\begin{align}
        \mathcal{K}_\mathcal{S} =\{&K=(E_u^\top E_u)^{-1}E_u^\top \tilde{Z}\tilde{Q}^{-1}E_x\in\mathcal{S}:\notag\\
        &\exists \tilde{Z}={\rm blkdiag}(\dots,\tilde{Z}_k,\dots),\notag\\ &\tilde{Z}_k\in\mathbb{R}^{m_{\mathcal{C}_k}\times n_{\mathcal{C}_k}},k\in\mathcal{Q}_\mathcal{G},\notag\\
        &\tilde{Q}={\rm blkdiag}(\dots,\tilde{Q}_k,\dots)\succ O,\notag\\ &\tilde{Q}_k\in\mathbb{R}^{n_{\mathcal{C}_k}\times n_{\mathcal{C}_k}},k\in\mathcal{Q}_\mathcal{G},\notag \\
        &\rho\in\mathbb{R}, \ \eta > 0,\notag\\
        &{\rm{s.t.}} \ 
        \begin{bmatrix}
                \tilde{Q}&\ast\\
                \tilde{A}_x\tilde{Q}+\tilde{B}_{xu}\tilde{Z}&\tilde{Q}
        \end{bmatrix}
        +\rho W_x\succ O,\notag \\ 
        & \hspace{1.6em} U_x\tilde{Q}+\tilde{Q}U_x \succeq \eta U_x\},
\end{align} 
where
\begin{align}
    W_x=&{\rm blkdiag}(U_x,U_x), \ U_x = I-E_x^\top(E_x^\top E_x)^{-1}E_x,\\
    \tilde{A}_x =& E_x A (E_x^\top E_x)^{-1} E_x^\top, \ \tilde{B}_{xu} = E_x B (E_u^\top E_u)^{-1} E_u^\top.
\end{align} Then $\mathcal{K}_{\rm diag} \subset \mathcal{K}_\mathcal{S} \subset \mathcal{K}_{\rm all}$ holds.

Here, $\tilde{Z}\tilde{Q}^{-1}$ has the same blkdiag structure as $\tilde{Z}$, allowing the transformation using Lemma 3'.
Since $(E_u^\top E_u)^{-1}$ is a diagonal matrix, $(E_u^\top E_u)^{-1}E_u^\top \tilde{Z}\tilde{Q}^{-1}E_x\in\mathcal{S}$ holds.
The proof of Lemma 5' and modifications of other lemmas and theorems can be obtained by appropriately adjusting the notations.\\

\subsection{Proof of Lemma~5}\label{App: A}
The proof can be done similarly to Theorem 1 and Theorem 2 in \cite{watanabe2024convex}.
Before proving Lemma \ref{lem: 1 linear}, we give the preliminary lemma, which is the nonconvex version of Lemma \ref{lem: 1 linear}.
\begin{lem}
        Let
        \begin{align}
                \hat{\mathcal{K}}_\mathcal{S}=\{&K=(E^\top E)^{-1}E^\top \tilde{Z}\tilde{Q}^{-1}E\notag\\
                 &=(E^\top E)^{-1}\sum_{k=1}^{q}E_{\mathcal{C}_k}^\top (\tilde{Z}_k\tilde{Q}_k^{-1}) E_{\mathcal{C}_k}:\notag\\
                 &\exists \tilde{Z}={\rm{blkdiag}}(\tilde{Z}_1,\cdots\tilde{Z}_q),\notag\\
                 &\tilde{Q}={\rm{blkdiag}}(\tilde{Q}_1,\cdots\tilde{Q}_q)\succ O, \rho\in\mathbb{R}\notag\\
                 &{\rm{s.t.}} \ 
                 \Phi (\tilde{Q}, \tilde{Z})
                 +\rho \tilde{R}W\tilde{R}\succ O\}
        \end{align}
        with $W$ in \eqref{FinsM},
        \begin{align}
                \Phi (\tilde{Q}, \tilde{Z})=
                \begin{bmatrix}
                        \tilde{Q}&\ast\\
                        \tilde{A}\tilde{Q}+\tilde{B}\tilde{Z}&\tilde{Q}
                \end{bmatrix}
        \end{align}
        and
        \begin{align}
                \tilde{R} \ \ \ &={\rm blkdiag}(\tilde{Q},\tilde{Q}).
        \end{align}
        Then, $\hat{\mathcal{K}}_\mathcal{S}\subset \mathcal{K}_{\rm all}$ holds.
\end{lem}
\begin{proof}
        Let $\tilde{K}=\tilde{Z}\tilde{Q}^{-1}$.
        For $(E^\top E)^{-1}E^\top \tilde{Z}\tilde{Q}^{-1}E \in \hat{\mathcal{K}}_\mathcal{S}$, it can be seen that
        \begin{align*}
                O\prec &\tilde{R}^{-1}(\Phi (\tilde{Q}, \tilde{Z})+\rho \tilde{R}M\tilde{R})\tilde{R}^{-1}\\
                = &
                \begin{bmatrix}
                        \tilde{Q}^{-1}&\ast\\
                        \tilde{Q}^{-1}(\tilde{A}+\tilde{B}\tilde{K})&\tilde{Q}^{-1}
                \end{bmatrix}
                +\rho M
        \end{align*}
        From Finsler's lemma, by setting $K = (E^\top E)^{-1}E^\top \tilde{K}E$ and $P=E^\top \tilde{Q}^{-1}E$, we get the following inequality:
        \begin{align*}
                O&\prec 
                \begin{bmatrix}
                        \ \hspace{-0.5em}E^\top \  O\ \\
                        O\  \ E^\top\
                \end{bmatrix}
                \begin{bmatrix}
                        \tilde{Q}^{-1}&\ast\\
                        \tilde{Q}^{-1}(\tilde{A}+\tilde{B}\tilde{K})&\tilde{Q}^{-1}
                \end{bmatrix}
                \begin{bmatrix}
                        E \ \ O \\
                        O\ \  E
                \end{bmatrix}\\
                &=
                \begin{bmatrix}
                        P&\ast\\P(A+BK)&P
                \end{bmatrix}
                =
                \begin{bmatrix}
                        P&\ast\\A+BK&P^{-1}
                \end{bmatrix}
        \end{align*}
        Here, Lemma \ref{lem: StoBlkdiag} guarantees $K\in \mathcal{S}$, thus $\hat{\mathcal{K}}_\mathcal{S}\subset \mathcal{K}_{\rm all}$.
\end{proof}

Finally, we prove Lemma~5. This is a convexification of Lemma 10.

\begin{proof}(Proof of Lemma~5)
        First, we show $\mathcal{K}_\mathcal{S} \subset \hat{\mathcal{K}}_\mathcal{S}$. 
        By combining two LMIs in \eqref{eq: thm2}, we get $\Phi (\tilde{Q}, \tilde{Z})+\rho/\eta (W\tilde{R}+ \tilde{R}W)\succ O$.
        Pre- and post-multiplying with $\tilde{R}$ and utilizing Finsler's lemma yields $\exists\beta \in \mathbb{R},\ \tilde{R}^{-1}\Phi (\tilde{Q}, \tilde{Z})\tilde{R}^{-1}+\beta W\succ O$.
        This satisfies a matrix inequality condition of $\hat{\mathcal{K}}_\mathcal{S}$, thus $\mathcal{K}_\mathcal{S} \subset \hat{\mathcal{K}}_\mathcal{S}$.

        Next, we prove $\mathcal{K}_{\rm diag} \subset \mathcal{K}_\mathcal{S}$. 
        For $K \in \mathcal{K}_{\rm diag}$, we have $K=ZQ^{-1}\in\mathcal{S}$ with some $Z\in\mathcal{S}$ and $Q={\rm blkdiag}\{Q_1,\dots,Q_N\}\succ O$.
        From lemma \ref*{lem: StoBlkdiag}, block-diagonal matrix $\tilde{K}$ such that $K=(E^\top E)^{-1}E^\top \tilde{K}E$ exists.
        We define $\tilde{Q}={\rm blkdiag}(\cdots,\tilde{Q}_k,\cdots)\succ O$ with $\tilde{Q}_k = {\rm blkdiag}(\cdots, \underset{j\in\mathcal{C}_k}{|\mathcal{Q}_\mathcal{G}^j|}Q_j,\cdots)\succ O$.
        By setting $\tilde{Z}=\tilde{K}\tilde{Q}$, we get $(E^\top E)^{-1} E^\top \tilde{Z}\tilde{Q}^{-1} E =\ (E^\top E)^{-1} E^\top \tilde{Z}E(E^\top E)^{-1}Q^{-1}=\ ZQ^{-1}=K$
        since $\tilde{Q}^{-1} E=E(E^\top E)^{-1}Q^{-1}$ and $Z$ can be represented as $Z=(E^\top E)^{-1} E^\top \tilde{Z}\tilde{Q}^{-1} E\in\mathcal{S}$ from lemma \ref*{lem: StoBlkdiag}.
        \begin{align*}
                O\prec&
                \begin{bmatrix}
                        Q^{-1}&\ast\\Q^{-1}(A+BK)&Q^{-1}
                \end{bmatrix}\\
                =&
                \begin{bmatrix}
                        E^\top \tilde{Q}^{-1}E&\ast\\E^\top \tilde{Q}^{-1}E(A+BK)&E^\top \tilde{Q}^{-1}E
                \end{bmatrix}\\
                =&
                \begin{bmatrix}
                        \ \hspace{-0.5em}E^\top \  O\ \\
                        O\  \ E^\top\
                \end{bmatrix}
                \begin{bmatrix}
                        \tilde{Q}^{-1}&\ast\\\tilde{A}+\tilde{B}\tilde{K}&\tilde{Q}^{-1}
                \end{bmatrix}
                \begin{bmatrix}
                        E \ \ O \\
                        O\ \  E
                \end{bmatrix}.
        \end{align*}
        By using Finsler's lemma, and utilizing $\tilde{Q}U=U\tilde{Q}$, we get
        \begin{align*}
                O\prec& \
                \Phi (\tilde{Q}, \tilde{Z})+\rho \tilde{R}W\tilde{R}\\
                =& \
                \Phi (\tilde{Q}, \tilde{Z})+\rho W\tilde{R}^2W\\
                \preceq& \
                \Phi (\tilde{Q}, \tilde{Z})+ \mu W,
        \end{align*}
        where $\mu = |\rho| \lambda_{\rm max}^2(\tilde{R})$.
        Moreover, $U\tilde{Q}+\tilde{Q}U=U^2\tilde{Q}+\tilde{Q}U^2=2U\tilde{Q}U\succeq 2\lambda_{\rm min}(\tilde{Q})U$.
        Here, $\lambda_{\rm max/min}(A)$ denote the maximum/minimum eigenvalue of $A$.
        Therefore, we obtain $K \subset \mathcal{K}_\mathcal{S}$ and thus $\mathcal{K}_{\rm diag} \subset \mathcal{K}_\mathcal{S}$.
\end{proof}

\subsection{SDP variables for the proposed method}\label{App: B}
Since the computational time can vary depending on solvers, we compare the number of variables involved in the semidefinite programming (SDP) formulations.
In the case that the edge set $\mathcal{E}$ is chosen as $\mathcal{Q}_\mathcal{G}$, the orders of the number of SDP variables are $\mathcal{O}(nd)$, $\mathcal{O}(nd)$, $\mathcal{O}(n^2)$ and $\mathcal{O}(n^2d^2)$ for $\mathcal{K}_{\rm diag}$,$\mathcal{K}_\mathcal{S}$, $\mathcal{K}_{\rm ext}$ and $\mathcal{K}_{\mathcal{S}, {\rm ext}}$, respectively, where $d$ denotes the maximum number of neighboring agents;
precisely,
\begin{align*}
  \mathcal{K}_{\rm diag}:&\quad \mathcal{F}(Z) \leq n(d+1)\hspace{5.2em} \\
  \mathcal{K}_\mathcal{S}:&\quad \mathcal{F}(\tilde{Z}) \leq 4nd\hspace{7.2em}   \\
  \mathcal{K}_{\rm ext}:&\quad \mathcal{F}(Q) \leq \frac{1}{2}n(n+1)\hspace{4.4em}   \\
  \mathcal{K}_{\mathcal{S}, {\rm ext}}:&\quad \mathcal{F}(\tilde{Q}) \leq \frac{1}{2}(2nd)(2nd+1) \hspace{1.5em}   
\end{align*}
where $\mathcal{F}(X)$ denotes the number of SDP variables in matrix $X$.
Here, the number of SDP variables of the proposed method $\mathcal{K}_{\mathcal{S}, {\rm ext}}$, $\mathcal{O}(n^2d^2)$, is the largest.
However, for a large, sparse graph satisfying $n\gg d$, the number approximates $\mathcal{O}(n^2)$, the same as the number of SDP variables of the extended LMIs, $\mathcal{K}_{\rm ext}$.

\subsection{Proof of Theorem 2}\label{App: C}
The proof can be done similarly to Theorem~\ref{thm: ext linear}, by considering the following lemma, a nonconvex version of Theorem~\ref{thm: ext linear}.
\begin{lem} 
        Let
        \begin{align}
          \hat{\mathcal{K}}_{\mathcal{S},{\rm ext}}^{\infty, \gamma}=\{&K=(E^\top E)^{-1}E^\top (\tilde{Z}\tilde{G}^{-1})E:\ \exists \tilde{Q}\succ O,\notag\\
          &\tilde{Z}={\rm blkdiag}(\tilde{Z}_1,\cdots\tilde{Z}_q)\notag\\
          &\tilde{G}={\rm blkdiag}(\tilde{G}_1,\cdots\tilde{G}_q), \ \rho \in \mathbb{R} \ {\rm s.t.}\notag\\
          &\Theta_\gamma(\tilde{G},\tilde{Q},\tilde{Z})+{\rm blkdiag}(\rho \tilde{H}W\tilde{H}, O_{2n\times 2n})\succ O\}
        \end{align}
        with $W$ in \eqref{FinsM}, $\tilde{H}$ in \eqref{eq: H} and $\Theta_\gamma$ in \eqref{Theta}.
        Then, $\hat{\mathcal{K}}_{\mathcal{S},{\rm ext}}^{\infty, \gamma}\subset\mathcal{K}_{\rm all}^{\infty, \gamma}$ holds.
\end{lem}
\end{document}